\documentclass[pra,twocolumn,aps,10pt,showpacs,superscriptaddress,letterpaper]{revtex4-1}

\usepackage{amsmath,amssymb,amsfonts,amsthm}
\usepackage{revsymb}
\usepackage{graphicx}
\usepackage{hyperref}

\newcommand{\calC}{{\cal C }}
\newcommand{\calE}{{\cal E }}
\newcommand{\calT}{{\cal T }}
\newcommand{\calZ}{{\cal Z }}
\newcommand{\calP}{{\cal P }}
\newcommand{\calG}{{\cal G }}
\newcommand{\calF}{{\cal F }}
\newcommand{\calK}{{\cal K }}

\newcommand{\FF}{\mathbb{F}}
\newcommand{\trace}{\mathop{\mathrm{Tr}}\nolimits}
\newcommand{\css}[1]{\mathrm{CSS}\left({#1}\right)}
\newcommand{\ket}[1]{\left| {#1} \right\rangle}
\newcommand{\bra}[1]{\left\langle {#1} \right|}

\newtheorem{dfn}{Definition}
\newtheorem{lemma}{Lemma}

\begin{document}

\title{Magic state distillation with low overhead}

\author{Sergey \surname{Bravyi}}
\affiliation{IBM Watson Research Center, Yorktown Heights,  NY 10598}

\author{Jeongwan \surname{Haah}}
\affiliation{Institute for Quantum Information and Matter, California Institute of Technology, Pasadena, CA 91125}

\date{12 September 2012}

\begin{abstract}
We propose a new family of error detecting stabilizer codes with an encoding rate $1/3$
that permit a transversal implementation of the gate $T=\exp{(-i\pi Z/8)}$ on all logical qubits.
The new codes are used to construct protocols for distilling high-quality `magic' states
$T \ket +$ by Clifford group gates and Pauli measurements.
The distillation overhead scales as $O(\log^\gamma{(1/\epsilon)})$,
where $\epsilon$ is the output accuracy and $\gamma=\log_2{(3)}\approx 1.6$.
To construct the desired family of codes, we introduce the notion of a triorthogonal matrix
--- a binary matrix in which any pair and any triple of rows have even overlap.
Any triorthogonal matrix gives rise to a stabilizer code with a transversal $T$-gate on all logical qubits, possibly augmented by Clifford gates.
A powerful numerical method for generating triorthogonal matrices is proposed.
Our techniques lead to a two-fold overhead reduction for distilling magic states
with accuracy $\epsilon \sim 10^{-12}$ compared with the best previously known protocol.
\end{abstract}


\maketitle

\section{Introduction}

Quantum error correcting codes provide a means of
trading quantity for quality when unreliable components
must be used to build a reliable quantum device.
By combining together sufficiently  many unprotected noisy qubits
and exploiting their collective degrees of freedom
insensitive to local errors,  quantum coding allows one to
simulate noiseless logical qubits and quantum gates up to any
desired precision provided that the noise level is below a constant threshold value~\cite{Dennis01,Shor96,Knill04,AGP06}.
Protocols for fault-tolerant quantum computation with the
error threshold close to $1\%$  have been proposed recently~\cite{Knill05,RH:cluster2D,Fowler08}.

An important figure of merit of fault-tolerant protocols is the cost
of implementing a given  logical operation such as a unitary gate or a measurement
with a desired accuracy $\epsilon$.  Assuming that elementary operations
on unprotected qubits
have unit cost, all fault-tolerant protocols proposed so far including the ones based on concatenated codes~\cite{AGP06} and topological codes~\cite{RH:cluster2D,RHG07,Fowler08}
enable implementation of a universal set of logical operations
with the cost $O(\log^\beta{(1/\epsilon)})$,
where the scaling exponent $\beta$ depends on a particular protocol.

For protocols based on stabilizer codes~\cite{Gottesman97}
the cost of a logical operation may also depend on
whether the operation is a {\em Clifford} or a {\em non-Clifford} one.
The set of Clifford operations (CO) consists of
unitary Clifford group gates such as the Hadamard gate $H$,
the $\pi/4$-rotation $S=\exp{(i\pi Z/4)}$,
and the CNOT gate,
preparation of ancillary $\ket 0$ states, and measurements in the
$\ket 0, \ket 1$ basis. Logical CO usually have a
relatively low cost as they can be implemented
either transversally~\cite{Gottesman97} or, in the case of topological
stabilizer codes, by the code deformation method~\cite{RHG07,BMD:codedef,Fowler08}.
On the other hand, logical non-Clifford
gates, such as the $\pi/8$-rotation $T=\exp{(-i\pi Z/8)}$
usually lack a transversal implementation~\cite{EastinKnill2009,BravyiKoenig12}
and have a relatively high cost that may exceed the one of CO
by orders of magnitude~\cite{RHG07}. Reducing the cost of
non-Clifford gates is an important problem since
the latter constitute a significant fraction of any
interesting quantum circuit.

The present paper addresses this problem by constructing
low overhead protocols for the magic state distillation --- a
particular method of implementing logical non-Clifford
gates proposed in~\cite{BK04}. A magic state is an ancillary resource
state $\psi$ that
combines two properties: \\
{\em Universality:} Some non-Clifford
unitary gate can be implemented using one copy of $\psi$
and CO. The ancilla $\psi$ can be destroyed in the
process. \\
{\em Distillability:}  An arbitrarily good approximation to $\psi$ can be prepared by
CO, given a supply of raw ancillas $\rho$ with the initial
fidelity  $\bra \psi \rho \ket \psi$ above some constant
threshold value. \\
Since the Clifford group augmented by any non-Clifford gate is computationally
universal~\cite{Nebe00}, magic state distillation can be used to achieve
universality at the logical level provided that logical CO and logical raw ancillas $\rho$ are readily available.

Below we shall focus on the magic state
\[
\ket A = T \ket{+} \sim \ket{0} + e^{i\pi/4} \ket{1} .
\]
A single copy of $\ket A$ combined with a few CO
can be used to implement  the $T$-gate, whereby providing
a computationally universal set of gates~\cite{Boykin00,BK04}.
It was shown by Reichardt~\cite{Reichardt05} that the state $\ket A$
is distillable if and only if the initial fidelity $\bra A \rho \ket A$ is above the threshold value
$(1+1/\sqrt{2})/2\approx 0.854$.

Our main objective will be to minimize
the number of raw ancillas $\rho$
required to distill magic states $\ket A$  with a desired accuracy $\epsilon$.
To be more precise,  let $\sigma$ be a state of $k$ qubits which
is supposed to approximate
$k$ copies of $\ket A$. We will say that $\sigma$ has an
{\em error rate } $\epsilon$ iff the marginal state
of any qubit has an overlap at least $1-\epsilon$ with $\ket A$.
Suppose such a state $\sigma$ can be prepared
by a   distillation protocol  that takes as input
$n$ copies of the raw ancilla $\rho$
and uses only CO. We will say that the protocol
has a {\em distillation cost} $C=C(\epsilon)$ iff $n\le Ck$.
For example, the original distillation protocol of Ref.~\cite{BK04}
based on the $15$-qubit Reed-Muller code has a  distillation cost
$O(\log^\gamma{(1/\epsilon)})$, where
$\gamma=\log_3{(15)}\approx 2.47$.

\section{Summary of results}
\label{sec:results}

Our main result is a new family of distillation protocols for the state $\ket A$
with a distillation cost $O(\log^\gamma{(1/\epsilon)})$,
where  $\gamma=\log_2{\left(\frac{3k+8}{k}\right)}$
and $k$ is an arbitrary even integer. By choosing large enough $k$
the scaling exponent $\gamma$
can be made arbitrarily close to $\log_2{(3)}\approx 1.6$.
The protocol works by concatenating an elementary subroutine
that takes as input $3k+8$ magic states with an error rate $p$
and outputs $k$ magic states with an error rate $O(p^2)$.
For comparison, the best previously known
protocol found by Meier et al.~\cite{MEK} has a distillation cost as above with
the scaling exponent $\gamma=\log_2{(5)}\approx 2.32$.
Distillation protocols with the scaling exponent $\gamma=2$
were recently discovered by Campbell et al.~\cite{Campbell12}
who studied extensions of stabilizer codes, CO,  and magic states to qudits. We conjecture that the scaling exponent $\gamma$ cannot be smaller than $1$
for {\em any} distillation protocol and give some arguments in support of this
conjecture in Section~\ref{sec:full}.

Our distillation scheme borrows two essential ideas from Refs.~\cite{BK04,MEK}.
First, as proposed in~\cite{BK04}, we employ stabilizer codes that admit a special symmetry in favor of transversal $T$-gates and measure the syndrome of such codes
to detect errors in the input magic states.
Secondly, as proposed by Meier et al.~\cite{MEK}, we reduce the distillation cost
significantly by using distance-$2$ codes with multiple logical qubits.
The new ingredient is a systematic method of constructing stabilizer
codes with the desired
properties. To this end we introduce the notion of a triorthogonal matrix ---
a binary matrix in which any pair and any triple of rows have even overlap.
We show that any triorthogonal matrix $G$ with $k$ odd-weight rows
can be mapped to a stabilizer code with $k$ logical qubits that admit a transversal $T$-gate on all logical qubits, possibly augmented by Clifford gates.
Each even-weight row of $G$ gives rise to a stabilizer which is used in the distillation protocol to detect errors in the input magic states.
Finally, we propose a powerful numerical method for generating triorthogonal matrices.
To illustrate its usefulness,
we construct the first example of a distance-$5$ code with a transversal $T$-gate
that encodes one qubit into $49$ qubits.

While the asymptotic scaling of the distillation cost is of great theoretical interest,
its precise value in the non-asymptotic regime may offer valuable insights
on practicality of a given protocol.  Using raw ancillas with the initial
error rate $10^{-2}$ and the target error rate $\epsilon$ between $10^{-3}$ and $10^{-30}$
we computed the distillation cost $C(\epsilon)$ numerically for the optimal
sequence composed of the $15$-to-$1$ protocol of Ref.~\cite{BK04},
and the $10$-to-$2$ protocol of Ref.~\cite{MEK}.  Combining these protocols
with the ones discovered in the present paper we observed
a two-fold reduction of the distillation cost for $\epsilon=10^{-12}$
and a noticeable cost reduction for the entire range of $\epsilon$,
see  Table~\ref{tb:cost} in Section~\ref{sec:cost}.

Since a magic state distillation is meant to be performed
at the logical level of some stabilizer code, throughout this paper
we assume that CO themselves are perfect. Whether or not this simplification is justified depends on the chosen code. More precisely, let the cost of implementing logical CO
 and the distillation cost be $\log^{\beta}(1/\epsilon)$ and $\log^{\gamma}(1/\epsilon)$
respectively, where $\epsilon$ is the desired precision. In the case $\beta<\gamma$,
high-quality CO are cheap and one can safely assume that CO are perfect.
The opposite case when high-quality CO are expensive (i.e. $\beta>\gamma$) is realized, for example, in the topological one-way quantum computer
based on the 3D cluster state introduced by Raussendorf et al.~\cite{RHG07},
where $\beta=3$. As was pointed out in~\cite{RHG07}, in this case
it is advantageous to use expensive high-quality CO only at the final rounds of distillation and use relatively cheap noisy CO for the initial rounds.
Using the $15$-to-$1$ distillation protocol of Ref.~\cite{BK04}
with $\gamma=\log_3{15}\approx 2.47$,
the authors of Ref.~\cite{RHG07} showed how to implement a universal set
of logical gates with the cost $O(\log^{3}(1/\epsilon))$.
A detailed analysis of errors in logical CO was performed by
Jochym-O'Connor et al~\cite{Laflamme12}.

The rest of the paper is organized as follows.
We begin with the definition of triorthogonal matrices
and state their basic properties in Section~\ref{sec:ort}.
The correspondence between  triorthogonal
matrices and  stabilizer codes with a transversal $T$-gate
is described in Section~\ref{sec:codes}.
We introduce
our distillation protocols for the magic state $\ket A$
in Sections~\ref{sec:dist},\ref{sec:full} and Appendix~\ref{appdst}.
A family of distance-$2$ codes with an encoding rate $1/3$
that admit a transversal $T$-gate
is presented in Section~\ref{sec:family}.
We compute the distillation cost of the new protocols and
make comparison with the
previously known protocols in Section~\ref{sec:cost}. A numerical method
of generating triorthogonal matrices is presented in Section~\ref{sec:linear}.
Finally, Appendix~\ref{app49} presents the $[[49,1,5]]$ code with a transversal $T$-gate.

{\em Notations:}  Below we adopt standard notations and terminology pertaining to
quantum stabilizer codes~\cite{NCbook}.
Given a pair of binary vectors $f,g\in \FF_2^n$, let
$(f,g)=\sum_{j=1}^n f_j g_j \pmod 2$ be their inner product
and $|f|$ be the weight of $f$, that is, the number of non-zero entries in $f$.
Given a linear space $\calG\subseteq \FF_2^n$, its dual space
$\calG^\perp$ consists of all vectors $f\in \FF_2^n$ such that $(f,g)=0$ for any
$g\in \calG$. We shall use notations $X,Y,Z$ for the single-qubit Pauli
operators.  Given any single-qubit operator $O$ and a binary vector
$f\in \FF_2^n$, the tensor product $O^{f_1}\otimes \cdots \otimes O^{f_n}$
will be denoted $O(f)$. In particular, $X(f) Z(g)=(-1)^{(f,g)} Z(g) X(f)$.
The Pauli group $\calP_n$ consists of
$n$-qubit Pauli operators $i^\omega \, P_1\otimes \cdots\otimes  P_n$,
where $P_j\in \{I,X,Y,Z\}$, and $\omega\in \mathbb{Z}_4$.
The Clifford group $\calC_n$ consists of all unitary operators $U$
such that $U\calP_n U^\dag =\calP_n$. It is well known that $\calC_n$
is generated by one-qubit gates $H=(X+Z)/\sqrt{2}$ (the Hadamard gate),
$S=\exp{(i\pi Z/4)}$ (the $S$-gate), and the controlled-$Z$ gate $\Lambda(Z)
=\exp{(i\pi \ket{11} \bra{11})}$.
All quantum codes discussed in this paper
are of Calderbank-Shor-Steane (CSS) type~\cite{CSS1,CSS2}.
Given a pair of linear spaces $\calF,\calG\subset \FF_2^n$ such that
$\calF\subseteq \calG^\perp$, the corresponding CSS code
has stabilizer group $\{ X(f)Z(g), \quad f\in \calF,\; g\in \calG\}$
and will be denoted as $\css{X,\calF;Z,\calG}$.

\section{Triorthogonal matrices}
\label{sec:ort}

To describe our distillation protocols let us define
a new class of binary matrices.
\begin{dfn}
A binary matrix $G$ of size $m\times n$ is called triorthogonal iff
the supports of any pair
and any triple of its rows
have even overlap, that is,
\begin{equation}
\label{ort2}
\sum_{j=1}^n G_{a,j} G_{b,j} = {0\pmod 2}
\end{equation}
for all pairs of rows $1\le a<b\le m$
and
\begin{equation}
\label{ort3}
 \sum_{j=1}^n G_{a,j} G_{b,j} G_{c,j} = {0\pmod 2}
\end{equation}
for all triples of rows $1\le a<b<c\le m$.
\end{dfn}
An example of a triorthogonal matrix of size $5\times 14$ is
\begin{equation}
\label{example1}
\setcounter{MaxMatrixCols}{14}
G=
\begin{bmatrix}
1    & 1   & 1   & 1   & 1   & 1    & 1   &     &     &     &     &     &     &     \\
     &     &     &     &     &      &     & 1   & 1   & 1   & 1   & 1   & 1   & 1   \\
\bf 1&     &\bf 1&     &\bf 1&      &\bf 1&\bf 1&     &\bf 1&     &\bf 1&     &\bf 1\\
     &\bf 1&\bf 1&     &     & \bf 1&\bf 1&     &\bf 1&\bf 1&     &     &\bf 1&\bf 1\\
     &     &     &\bf 1&\bf 1& \bf 1&\bf 1&     &     &     &\bf 1&\bf 1&\bf 1&\bf 1\\
\end{bmatrix},
\end{equation}
where only non-zero matrix elements are shown.
The two submatrices of $G$ formed by even-weight and odd-weight rows
will be denoted $G_0$ and $G_1$ respectively.
The submatrix $G_0$ is  highlighted in bold in Eq.~(\ref{example1}).
We shall always assume that $G_1$ consists of the first $k$ rows of $G$
for some $k\ge 0$.
Define linear subspaces $\calG_0,\calG_1,\calG\subseteq \FF_2^n$
spanned by the rows of $G_0$, $G_1$, and $G$ respectively.
Using Eq.~(\ref{ort2}) alone one can easily prove the following.
\begin{lemma}
\label{lemma:simple}
Suppose $G$ is triorthogonal. Then
(i) all rows of $G_1$ are linearly independent over $\FF_2$,
(ii) $\calG_0\cap \calG_1=0$,
(iii)  $\calG_0= \calG\cap \calG^\perp$,
and (iv) $\calG_0^\perp=\calG_1\oplus \calG^\perp$.
\end{lemma}
\begin{proof}
Let $f^1,\ldots,f^m$ be the rows of $G$ such that the first $k$ row
form $G_1$.
By definition, any vector
$f\in \calG_1$ can be written as
$f=\sum_{a=1}^k x_a f^a$ for some $x_a\in \FF_2$.
From Eq.~(\ref{ort2}) we infer that $(f^a,f^b)=\delta_{a,b}$ for all $1\le a,b\le k$
and $(f^a,g)=0$ for any $g\in \calG_0$.
Hence  $x_a=(f,f^a)$. If $f=0$ or $f\in \calG_0$ then $x_a=0$ for all $a$.
This proves (i) and (ii). Since any row of $G_0$ is orthogonal to itself and
any other row of $G$, we get $(f,g)=0$ for all $f\in \calG_0$ and $g\in \calG$.
This implies $\calG_0\subseteq  \calG\cap \calG^\perp$.
If $f=\sum_{a=1}^m x_a f^a\in \calG\cap \calG^\perp$, then
$x_a=(f,f^a)=0$ for all $1\le a\le k$, that is, $f\in \calG_0$.
This proves (iii). Finally, (iv) follows
from $\calG_1\oplus \calG^\perp \subseteq \calG_0^\perp$,
$\calG_1\cap \calG^\perp=0$, and dimension counting.
\end{proof}
As we show in Section~\ref{sec:codes}, any binary matrix $G$ with $n$ columns and $k$ odd-weight rows
satisfying Eq.~(\ref{ort2}) gives rise to a stabilizer code encoding $k$ qubits into $n$ qubits. Condition Eq.~(\ref{ort3}) ensures that this code has the desirable transversality
properties, namely,  the encoded $\ket{ A^{\otimes k} }$ state can be
prepared by applying the transversal $T$-gate $T^{\otimes n}$
to the encoded $\ket{+^{\otimes k}}$, possibly augmented by some Clifford
operator.
To state this more formally,  define $n$-qubit unnormalized states
\begin{equation}
\label{G0G}
\ket{G_0} =\sum_{g\in \calG_0} \ket g \quad \text{and} \quad \ket G =\sum_{g\in \calG} \ket g.
\end{equation}
Define also a state
\begin{equation}
\label{Ak}
| \overline{A^{\otimes k}} \rangle = \prod_{a=1}^k ( I+e^{i\pi/4} X(f^a) ) \, \ket{G_0},
\end{equation}
where $f_1,\ldots,f_k$ are the rows of $G_1$.
\begin{lemma}
\label{lemma:transversal}
Suppose a matrix $G$ is triorthogonal. Then there exists
a Clifford group operator $U$ composed of
$\Lambda(Z)$ and $S$ gates only such that
\begin{equation}
\label{encodedA}
|\overline{A^{\otimes k}}\rangle=UT^{\otimes n} \ket G.
\end{equation}
\end{lemma}
\begin{proof}
Below we promote the elements of binary field $\FF_2$ to the normal integers of $\mathbb{Z}$;
we associate $\FF_2 \ni 0 \mapsto 0 \in \mathbb{Z}$ and $\FF_2 \ni 1 \mapsto 1 \in \mathbb{Z}$.
Unless otherwise noted by ``$(\mathrm{mod}~2)$'' or ``$(\mathrm{mod}~4)$'',
every sum is the usual sum for integers and no modulo-reduction is performed.

When $y = (y_1,\ldots,y_m)$ is a string of $0$ or $1$,
let  $\epsilon(y) \equiv  {|y| \pmod 2}$ be the parity of $y$.
 Let us derive a formula for a phase factor $e^{i\pi \epsilon(y)/4}$
as a function of components $y_a$.
Observe that
\begin{equation}
\label{eq:parity}
\epsilon(y)=\frac{1}{2} \left(1- (1-2)^{|y|}\right)=\sum_{p=1}^{|y|} \binom{|y|}{p} (-2)^{p-1}.
\end{equation}
Since the binomial coefficient $\binom{|y|}{p}$ is the number of 
ways to choose $p$ non-zero components of $y$, 
we may write
\begin{align}
e^{i\pi\epsilon(y)/4}= \exp{\left[
\frac{i\pi}4 \sum_{a=1}^m y_a  \right. } & -\frac{i\pi}2 \sum_{a<b} y_a y_b \nonumber \\
& {\left. + i\pi \sum_{a<b<c} y_a y_b y_c\right]}. \label{eq:mod8}
\end{align}

By definition of the state $\ket G$, one has
\[
T^{\otimes n} \ket G = \sum_{f \in \calG} e^{i{\pi |f|}/4}\, \ket f.
\]
Since $\ket G$ depends on the linear space $\calG$ rather than
the matrix presentation $G$, we may assume that all rows of $G$ are linearly independent over $\FF_2$.
Let $g^1,\ldots,g^m$ be the rows of $G$, and
decompose $f= \sum_{a=1}^m x_a g^a \pmod 2$,
where $x_a \in \{0,1\}$ are uniquely determined by $f$.

Each component $f_j$ of $f$ is
the parity of the bit string $(x_1 g^1_j, x_2 g^2_j, \ldots, x_m g^m_j)$,
and $|f|$ is the sum of $f_j$'s. Hence, Eq.~\eqref{eq:mod8} implies
\begin{align}
e^{i\pi |f| /4}= \exp{\left[ \vphantom{\frac{i\pi}4 \sum_{a=1}^m x_a |g^a| } \right. } &
\frac{i\pi}4 \sum_{a=1}^m x_a |g^a|   -\frac{i\pi}2 \sum_{a<b} x_a x_b |g^a\cdot g^b| \nonumber \\
& {\left. + i\pi \sum_{a<b<c} x_a x_b x_c |g^a \cdot g^b\cdot g^c| \right]}, \label{eq:phase_factor}
\end{align}
where $g^a \cdot g^b$ denotes the bitwise AND operation.
Triorthogonality condition Eq.~(\ref{ort3}) implies that
the triple overlap $|g^a\cdot g^b\cdot g^c|$ is even, so 
we may drop the last term in Eq.~(\ref{eq:phase_factor}).
This is in fact one of the main motivations we consider triorthogonal matrices.

Let the first $k$ rows of $G$ have odd weight and all others even weight, and put
\[
|g^a| =
\begin{cases}
 2 \Gamma_a +1 & \text{if } 1 \le a \le k, \\
 2 \Gamma_a    & \text{otherwise.}
\end{cases}
\]
In addition, Eq.~\eqref{ort2} implies for distinct $a,b$ that
\[
 |g^a \cdot g^b| = 2 \Gamma_{ab}.
\]
Here all $\Gamma_a$ and $\Gamma_{ab}$ are integers.
Thus
\[
e^{i\pi |f| /4}=\exp{\left[ \frac{i\pi}4 \sum_{a=1}^k x_a \right]} \cdot
\exp{\left[  \frac{i\pi}2  Q(x_1,\ldots,x_m)\right]},
\]
where
\[
Q(x)=\sum_{a=1}^m \Gamma_a\, x_a -2 \sum_{a<b} \Gamma_{ab}\,  x_a x_b .
\]
Let us  show that the unwanted phase factor $e^{i\pi Q/2}$ can be
canceled by a unitary Clifford operator that uses only $\Lambda(Z)$ and $S$ gates.
To this end, we rewrite $Q(x)$ as a function of $f$.
As noted earlier, $x_a$ are uniquely determined by $f$.
Indeed, there is a matrix $B$ over $\FF_2$ such that $x_a = \sum_p B_{ap} f_p \pmod 2$,
since $\{ g^a \}$ is a basis of the linear space $\calG$.
(There could be many such $B$.)
We again use Eq.~\eqref{eq:parity} with the observation that
$x_a$ is the parity of the bit string $(B_{a1} f_1, \ldots, B_{an} f_n)$ to infer
\begin{align*}
x_a       &= \sum_p B_{ap} f_p - 2 \sum_{p < q} B_{ap}B_{aq} f_p f_q  & \pmod 4 ,\\
2 x_a x_b &= 2 \sum_{p,q} B_{ap} B_{bq} f_p f_q                       & \pmod 4
\end{align*}
for all $a,b=1,\ldots,m$.
Therefore, we can express $Q(x)$ as
\[
Q(x(f))=\sum_{p=1}^n \Lambda_p f_p - 2\sum_{p<q} \Lambda_{pq} f_p f_q \pmod 4,
\]
where $\Lambda_p, \Lambda_{pq}$ are some integers determined by $B, \Gamma_a$, and $\Gamma_{ab}$,
all of which depend only on our choice of the matrix $G$.
Explicitly, $\Lambda_p = \sum_a \Gamma_a B_{ap} - 2 \sum_{a<b} \Gamma_{ab} B_{ap}B_{bp}$
and $\Lambda_{pq} = \sum_a \Gamma_a B_{ap}B_{aq} - \sum_{a<b} \Gamma_{ab}(B_{ap}B_{bq} + B_{bp}B_{aq})$.

The extra phase factor $e^{i\pi Q/2}$ is canceled by
applying $\Lambda(Z)^{\Lambda_{pq}}$ gate for each pair of qubits $p<q$,
and the gate $(S^\dag)^{\Lambda_p}$ to every qubit $p$.
This defines the desired Clifford operator $U$ composed of $\Lambda(Z)$ and $S$ gates
such that
\begin{equation}
\label{UT}
UT^{\otimes n} \ket f = \exp{\left[ \frac{i\pi}4 \sum_{a=1}^k x_a \right]} \, \ket f
\end{equation}
for all $f=\sum_{a=1}^m x_a g^a \pmod 2 \in \calG$. Therefore,
\[
UT^{\otimes n} \ket G = \prod_{a=1}^k (I + e^{i\pi/4} X(g^a)) \ket{G_0} = |\overline{A^{\otimes k}} \rangle.
\]
\end{proof}

For the later use let us state the following simple fact.
\begin{lemma}
\label{lemma:3rows}
Let $G$ be a triorthogonal matrix without zero columns.
If  $G_1$ is non-empty and
$G_0$ has less than $3$ rows, then $G_0$ must have
at least one zero column.
\end{lemma}
\begin{proof}
Suppose on the contrary all columns of $G_0$ are nonzero.
If $G_0$ has only one row, it must be the all-ones vector $1^n$.
Then, the inner product between $1^n$ and any row $f$ of $G_1$ is
the weight of $f$ modulo 2, which is odd. But, the orthogonality Eq.~\eqref{ort2}
requires it to be even. This is a contradiction.

Suppose now that $G_0$ has two rows $g_1, g_2$.
By permuting the columns we may assume that
$G_0 = \begin{bmatrix} A & B & C \end{bmatrix}$ where
\[
 A = \begin{bmatrix} 1  & \cdots & 1 \\ 0 &  \cdots & 0 \end{bmatrix}, \,
 B = \begin{bmatrix} 0  & \cdots & 0 \\ 1 &  \cdots & 1 \end{bmatrix}, \,
 C = \begin{bmatrix} 1  & \cdots & 1 \\ 1 &  \cdots & 1 \end{bmatrix}.
\]
Choose an odd-weight row $f$ of $G_1$, and let $w_A, w_B, w_C$ be the weight
of $f$ restricted to the columns of $A,B,C$, respectively.
The (tri)orthogonality Eqs.~(\ref{ort2},\ref{ort3}) implies
\begin{align*}
|g_1 \cdot f| &= w_A + w_C = 0 & \pmod 2, \\
|g_2 \cdot f| &= w_B + w_C = 0 & \pmod 2, \\
|g_1 \cdot g_2 \cdot f| &= w_C = 0 &\pmod 2.
\end{align*}
This is a contradiction since $|f| = w_A + w_B + w_C = 1 \pmod 2$.
\end{proof}

\section{Stabilizer codes based on triorthogonal matrices}
\label{sec:codes}

Given a triorthogonal matrix $G$ with $k$ odd-weight rows,
define a stabilizer code $\css{X,\calG_0;Z,\calG^\perp}$
with $X$-type stabilizers $X(f)$, $f\in \calG_0$, and $Z$-type stabilizers
$Z(g)$, $g\in \calG^\perp$. The inclusion $\calG_0\subseteq \calG$ implies
that all stabilizers pairwise commute.
\begin{lemma}
\label{lemma:code}
The  code $\css{X,\calG_0;Z,\calG^\perp}$
has $k$ logical qubits. Its logical Pauli operators  can be chosen as
\begin{equation}
\label{logical}
\overline{X}_a= X(f^a) \quad \mbox{and} \quad \overline{Z}_a=Z(f^a), \quad a=1,\ldots,k,
\end{equation}
where $f^1,\ldots,f^k$ are the rows of $G_1$.
The states $\ket{G_0}$, $\ket G$, and $| \overline{A^{\otimes k}} \rangle$
defined in Eqs.~(\ref{G0G},\ref{Ak}) coincide with encoded
states $\ket{0^{\otimes k}}$, $\ket{+^{\otimes k}}$, and $\ket{ {A^{\otimes k}} }$ respectively.
\end{lemma}
\begin{proof}
Indeed, the assumption that
$f^a$ have odd weight and Eq.~(\ref{ort2}) ensure that the operators
defined in Eq.~(\ref{logical}) obey the correct commutation rules, that is,
$\overline{X}_a \, \overline{Z}_b = (-1)^{\delta_{a,b}}  \overline{Z}_b \, \overline{X}_a$.
It remains to check that $\overline{X}_a$ and $\overline{Z}_a$ commute with all
stabilizers.
Given any $Z$-type stabilizer $Z(g)$, $g\in \calG^\perp$, one has
$X(f^a)Z(g)=(-1)^{(f^a,g)} Z(g) X(f^a)=  Z(g) X(f^a)$ since $f^a\in \calG$ and $g\in \calG^\perp$. Given any $X$-type stabilizer $X(f)$, $f\in \calG_0$, one has
$Z(f^a)X(f)=(-1)^{(f^a,f)} X(f) Z(f^a)=  X(f) Z(f^a)$ since $f^a\in \calG$ and
$\calG_0\subseteq \calG^\perp$, see Lemma~\ref{lemma:simple}.
This shows that $\overline{X}_a$ and $\overline{Z}_a$ are indeed   logical Pauli operators on $k$ encoded qubits.

Property~(iii) of Lemma~\ref{lemma:simple} implies that
$Z(g)\, \ket f = \ket f$ for any $f\in \calG_0$ and any $g\in \calG+\calG^\perp$.
Thus the state $\ket{G_0}$ defined
 in Eq.~(\ref{G0G}) coincides with the encoded $\ket{0^{\otimes k}}$ state.
It follows that $\ket G = \prod_{a=1}^k (I+\overline{X}_a) \ket{G_0}$
is the encoded $\ket{+^{\otimes k}}$ state, while
$| \overline{A^{\otimes k}} \rangle = \prod_{a=1}^k (I+e^{i\pi/4}\overline{X}_a) \ket{G_0}$
is the encoded $\ket{A^{\otimes k}}$
(ignoring the normalization).
\end{proof}
Using Lemma~\ref{lemma:code} one can show that 
the operator $UT^{\otimes n}$
defined in Lemma~\ref{lemma:transversal} implements an encoded $T$ gate
on each logical qubit of the code $\css{X,\calG_0;Z,\calG^\perp}$.
Indeed, for any $x\in \FF_2^k$, the encoded state $\ket{x}\equiv \ket{x_1,\ldots,x_k}$ is
\[
\ket{\overline{x}}=\overline{X}_1^{x_1}\cdots \overline{X}_k^{x_k} \ket{G_0}
=\sum_{f\in {\cal G}_0+x_1 f^1 +\ldots +x_k f^k}\; \ket{f}.
\]
Using Eq.~(\ref{UT}) from the proof of Lemma~\ref{lemma:transversal}
one arrives at
\[
UT^{\otimes n} \, \ket{\overline{x}}=e^{i\frac{\pi}4 \sum_{a=1}^k x_a}  \, \ket{\overline{x}}.
\]
This provides a generalization of a transversal $T$-gate to multiple logical qubits.

\section{Distillation subroutine}
\label{sec:dist}

We are now ready to describe the elementary distillation subroutine.
It takes as input $n$ copies of a (mixed) one-qubit ancilla $\rho$ such that
$\bra A \rho \ket A = 1-p$. We shall refer to $p$ as the {\em input error rate}.
Define single-qubit basis states
$ \ket{A_0} \equiv \ket{A}$ and $\ket{A_1} \equiv Z \ket A$.
We shall assume that $\rho$ is
diagonal in the $A$-basis, that is,
\begin{equation}
\label{standard}
\rho=(1-p)\ket{A_0}\bra{A_0} + p \ket{A_1}\bra{A_1}.
\end{equation}
This can always be achieved by applying operators $I$ and $A\equiv e^{-i\pi/4}  SX$
with probability $1/2$ each to every copy of $\rho$.
Note that $A\, \ket{A_\alpha} =(-1)^\alpha \ket{A_\alpha}$,
that is, the random application of $A$ is equivalent to the dephasing in the $A$-basis
which destroys the off-diagonal matrix elements $\bra{A_0} \rho \ket{A_1} $
without changing the fidelity $\bra{A_0} \rho \ket{A_0}$.

Define linear maps
\begin{equation}
\label{TE}
\calT(\eta)=T\eta T^\dag \quad \text{and} \quad
\calE(\eta)=(1-p)\eta + p Z\eta Z
\end{equation}
describing the ideal $T$-gate and the $Z$-error respectively.
Using Clifford operations and one copy of $\rho$ as in Eq.~(\ref{standard})
one can implement a noisy version of the $T$-gate, namely,
$\calE\circ \calT$. A circuit implementing $\calE\circ \calT$
is shown on Fig.~\ref{fig:Tgate}, where the $Z$-error $\calE$
is shown by the $Z$-gate box  with a subscript $p$ indicating the error probability.
One can easily show that this circuit indeed implements $\calE\circ \calT$ by commuting $\calE$ through the CNOT gate and the classically controlled $SX$ gate.

\begin{figure}[h]
\centerline{\includegraphics[height=4cm]{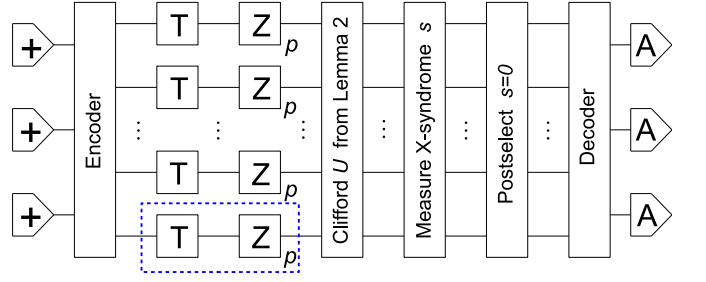}}
\caption{The distillation subroutine for the magic state $\ket A$
based on a triorthogonal matrix $G$.
The encoder prepares $k$ copies of the state $\ket +$
encoded by the stabilizer code $\css{X,\calG_0;Z,\calG^\perp}$.
Implementation of each $T$-gate consumes one ancillary
$\ket A$ state as shown on Fig.~\ref{fig:Tgate}. If the ancillas $\ket A$
have error rate $p$, each ideal $T$-gate is followed by a $Z$-error with
probability $p$. The Clifford operator $U$ is constructed in Lemma~\ref{lemma:transversal}.
Note that $U$ is diagonal in the $Z$-basis and thus
commutes with any $Z$-error. The syndrome $s$ is measured only
for $X$-type stabilizers $X(f^a)$ where $f^a$ are the rows of $G_0$.
In the case when all stabilizers $X(f^a)$ have eigenvalue $+1$ (trivial syndrome)
the decoder is applied.  It returns $k$ copies
of the state $\ket A$ with the overall error probability $O(p^d)$.
The trivial syndrome is observed with probability $1-O(p)$.
}
\label{fig:protocol}
\end{figure}
\begin{figure}[h]
\centerline{\includegraphics[height=2.2cm]{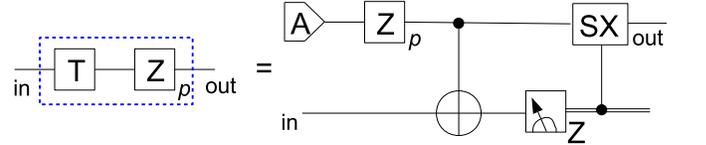}}
\caption{Implementation of the $T$-gate using CO
and one copy of the ancillary state $\ket A$.
If the ancilla
is a mixture of $\ket A$ and $Z \ket A$ with probabilities $1-p$ and $p$
respectively, the circuit enacts a noisy version of the $T$-gate, namely,
$\rho_{out}=(1-p) T \rho_\text{in} T^\dag + p ZT \rho_\text{in} T^\dag Z=\calE\circ \calT(\rho_\text{in})$.
The above circuit is used $n$ times in the subroutine of Fig.~\ref{fig:protocol}.
}
\label{fig:Tgate}
\end{figure}
The entire subroutine is illustrated on Fig.~\ref{fig:protocol}.
The first step is to prepare $k$ copies of the state
$\ket +$ and encode them using the code $\css{X,\calG_0;Z,\calG^\perp}$.
This results in the state $\ket G$ defined in  Eq.~(\ref{G0G})
and requires only CO.

The state $\ket G$ is then acted upon by the map $(\calE\circ \calT)^{\otimes n}$.
The latter can be implemented using CO and $n$ copies of
$\rho$ as shown on Fig.~\ref{fig:Tgate}. This results in a state
\[
\eta_1\equiv(\calE\circ \calT)^{\otimes n}\left( |G\rangle\langle G| \right)
=\calE^{\otimes n}\left( \hat{T} |G\rangle\langle G| \hat{T}^\dag \right),
\]
where $\hat{T}\equiv T^{\otimes n}$. Next we apply the Clifford unitary operator $U$
constructed in Lemma~\ref{lemma:transversal}. Since $U$ involves only $\Lambda(Z)$
and $S$ gates, it commutes with any $Z$-type error. Hence
the state prepared at this point is
\[
\eta_2\equiv U\eta U^\dag
= \calE^{\otimes n}\left( U \hat{T} |G \rangle\langle G| \hat{T}^\dag U^\dag \right)
= \calE^{\otimes n}\left( |\overline{A^{\otimes k}}\rangle \langle \overline{A^{\otimes k}}| \right),
\]
where we have used Eq.~(\ref{encodedA}).
The next step is a non-destructive eigenvalue measurement
for $X$-type stabilizers of the code $\css{X,\calG_0;Z,\calG^\perp}$,
that is, the Pauli operators $X(f^{k+1}),\ldots,X(f^m)$, where
$f^{k+1},\ldots,f^m$ are the rows of $G_0$. If at least one of the measurement
returns the outcome `$-1$', the subroutine returns `FAILED' and the final state
is discarded. If all measured eigenvalues are `$+1$',
the state $\eta_2$ has been projected onto the code space of the
code $\css{X,\calG_0;Z,\calG^\perp}$  and the subroutine is deemed successful
(since we do not have any $X$-type errors, the syndrome of all $Z$-type stabilizers is automatically trivial). This results in a state
\[
\eta_3=\Pi_0 \eta_2 \Pi_0/P_s,
\]
where $\Pi_0$ is the projector onto the code space of $\css{X,\calG_0;Z,\calG^\perp}$
and $P_s=\trace{(\eta_2 \Pi_0)}$ is the success probability.
The state $\eta_3$ has only contribution from errors
$Z(f)$ with $f\in \calG_0^\perp=\calG_1\oplus \calG^\perp$,
see Lemma~\ref{lemma:simple}, since these are the only $Z$-type errors
commuting with all $X$-type stabilizers. Hence the success probability is
\begin{equation}
\label{P_s}
P_s=\sum_{f\in  \calG_0^\perp} (1-p)^{n-|f|} p^{|f|}=\frac1{|\calG_0|} \sum_{f\in  \calG_0} (1-2p)^{|f|},
\end{equation}
where the second equality uses the MacWilliams identity~\cite{macslo}.
Any vector $f\in \calG_1\oplus \calG^\perp$
can be written as $f=g+x_1 f^1 + \ldots + x_kf^k$, where $g\in \calG^\perp$
and $f^1,\ldots,f^k$ are the rows of $G_1$. Since $Z(g)$ is a stabilizer, we conclude that
\begin{align*}
Z(f) | \overline{A^{\otimes k}} \rangle
&= Z(x_1 f^1 + \ldots + x_kf^k) | \overline{A^{\otimes k}} \rangle \\
&= \overline{Z}_1^{x_1} \cdots \overline{Z}_k^{x_k} | \overline{A^{\otimes k}} \rangle.
\end{align*}
Here we used definition of the logical $Z$-type operators, see Eq.~(\ref{logical}).
Hence the state $\eta_3$
 coincides  with an encoded $k$-qubit mixed state
\begin{equation}
\label{rho_out}
\rho_{out}=\frac1{P_s} \sum_{x\in \FF_2^k} p_{out}(x) \ket{A_x}\bra{A_x},
\end{equation}
where $\ket{A_x}=\ket {A_{x_1}} \otimes \cdots \otimes  \ket{A_{x_k}}$ and
\begin{equation}
\label{p_out}
p_{out}(x)=\sum_{f\in  \calG^\perp+ x_1 f^1 + \ldots + x_k f^k} \; \;
(1-p)^{n-|f|} p^{|f|}.
\end{equation}
The last step of the subroutine is to decode $\css{X,\calG_0;Z,\calG^\perp}$
whereby mapping $\eta_3$ to $\rho_{out}$.
The $k$-qubit state $\rho_{out}$ is
the output state of the distillation subroutine.
The reduced density matrix describing the $a$-th output qubit can
be written as
\[
\rho_{out,a}=(1-q_a)\ket{A_0}\bra{A_0} + q_a \ket{A_1}\bra{A_1},
\]
where $q_a$ is the output error rate on the $a$-th qubit:
\[
q_a=1-\frac1{P_s} \sum_{x\, :\, x_a=0} p_{out}(x).
\]
Let $\calK_a$ be the sum of $\calG^\perp$ and the space
spanned by all rows of $G_1$ except for $a$. Lemma~\ref{lemma:simple}
implies that $\dim{\calK_a}=\dim{\calG_0^\perp}-1$.
On the other hand, $\calK_a \subseteq (\calG_0 \oplus (f^a))^\perp$,
where $(f^a)=\{0^n,f^a\}$ is the one-dimensional subspace spanned by $f^a$.
Hence $\calK_a=(\calG_0 \oplus (f^a))^\perp$ and thus
\begin{equation}
\label{q}
q_a=1-\frac{\sum_{f\in  (\calG_0 \oplus (f^a))^\perp} (1-p)^{n-|f|} p^{|f|}}
{\sum_{f\in  \calG_0^\perp} (1-p)^{n-|f|} p^{|f|}}.
\end{equation}
We shall be mostly interested in the worst-case output error rate
\begin{equation}
\label{qmax}
q=\max_{a=1,\ldots,k} \; q_a.
\end{equation}
Output qubits with $q_a<q$ can be additionally dephased in the $A$-basis
to achieve $q_a=q$.
From Eq.~(\ref{q}) we infer that
$q=O(p^d)$, where $d$ is the minimum weight of a vector $f\in \calG_0^\perp$
such that $(f,f^a)=1$ for some $a$. Equivalently,
\begin{equation}
\label{distance}
d=\min_{f\in \calG_0^\perp\backslash \calG^\perp}\; |f|
\end{equation}
is the distance of the code $\css{X,\calG_0;Z,\calG^\perp}$
against $Z$-type errors.
Using the MacWilliams identity, we also get
\begin{equation}
\label{qdual}
q_a=1-\frac12 \frac{ \sum_{f\in  \calG_0\oplus (f^a)} (1-2p)^{|f|}}
{\sum_{f\in  \calG_0} (1-2p)^{|f|}}.
\end{equation}
This expression can be easily evaluated numerically
in the important case when $G_0$ has only a few rows.

The above subroutine requires $n$ extra
qubits to prepare the encoded $\ket{+^{\otimes k}}$ state,
while the total number of Pauli measurements is $n+m-k$.
In Appendix~\ref{appdst} we describe an alternative subroutine
which is slightly less intuitive but does not require any extra qubits
and uses only $n-k$ Pauli measurements.
Both subroutines output the same state and have the same success
probability.

\section{Full distillation protocol}
\label{sec:full}

The final goal of the distillation is to prepare a state $\sigma$ of $N$ qubits
such that the overlap between $\sigma$ and $N$-copies of the magic
state $\ket A$ is sufficiently close to $1$, say, at least $2/3$.
Such state $\sigma$ can be used as a resource to simulate
any quantum circuit that contains Clifford gates and at most $N$ gates $T$
using only CO with an overall  error probability at most $1/3$.
Each qubit of $\sigma$ allows one to simulate one $T$-gate using the scheme
shown on Fig.~\ref{fig:Tgate}.

Let $\sigma_j$ be the reduced density matrix describing the $j$-th qubit
of $\sigma$. For any given target error rate $\epsilon$ our full protocol will distill a state
$\sigma$ which is diagonal in the basis $\{ \ket{A_0}, \, \ket{A_1} \}^n$ and such that
\begin{equation}
\label{marginal}
\max_j \bra{A_1} \sigma_j \ket{A_1} \le \epsilon.
\end{equation}
The standard union bound then implies that the overlap
$\bra{ A_0^{\otimes N} } \sigma \ket{ A_0^{\otimes N}}$ is close to $1$
whenever $\epsilon \sim 1/N$.

In order to distill $N$ magic states with the target error rate $\epsilon$, the elementary
subroutine described in Section~\ref{sec:dist} will be applied recursively
such that each input state $\rho$ consumed by a level-$m$ distillation
subroutine is one of the output states $\rho_{out,a}$ distilled by some  level-$(m-1)$
subroutine. The recursion starts at a level $m=0$ with $NC$ input states,
where $C=C(\epsilon)$ is the distillation cost.
 In the limit $N \gg 1$ the distillation rounds can be organized such
that all $n$ input states $\rho$ consumed by any elementary subroutine at a level $m$
have been distilled at {\em different} subroutines at the level $m-1$, see Lemma~IV in~\cite{MEK}.
It allows one to disregard correlations between errors  and analyze
the full protocol using the average {\em yield}
\[
\Gamma(p)=\frac{k P_s(p)}{n},
\]
that is, the average number of output states with an error rate $q(p)$ per one input state
with an error rate $p$.  Here
$q$ is defined in Eqs.~(\ref{qmax},\ref{qdual}).
Neglecting the fluctuations, the distillation cost $C$,
the input error rate $p$,
the target error rate $\epsilon$, and the required number of levels $m_0$ are related by the following obvious equations:
\begin{align}
C_{m+1}&=\Gamma(p_m)C_m,           &\nonumber \\
p_{m+1}&= q(p_m),                  & m=0,\ldots,m_0-1, \nonumber \\
p_{m_0}&=\epsilon, \quad p_0=p,    & \nonumber \\
C_{m_0}&=1, \quad C_0=C.
\label{eq:cost-recursion}
\end{align}
In the limit of small $p$ one has $P_s(p)\approx 1$ and thus $\Gamma(p)\approx k/n$.
Taking into account that $q=O(p^d)$, where the distance $d$ is defined in Eq.~(\ref{distance}), one arrives at
\begin{equation}
\label{gamma}
C(\epsilon)=O(\log^\gamma{(1/\epsilon)}), \quad \gamma=\frac{\log{(n/k)}}{\log{(d)}},
\end{equation}
provided that the input error rate $p$ is below a constant threshold value $p_\text{th}$,
that depends on the chosen triorthogonal matrix.

We conjecture that the scaling exponent $\gamma$ of the
distillation cost $C$ cannot be smaller than $1$
for any concatenated distillation protocol based on a triorthogonal matrix.
Indeed, suppose the output error rate satisfies $q(p) \le c p^d < p$ for $p < p_0$ and $q(1) =1$.
As noted above, the potential correlation in the error probabilities among the output states may be ignored.
Then, after $m$ levels of distillation the output error rate should satisfy
\[
\epsilon \le c^{-1/(d-1)}(c'p_0)^{d^m}
\]
where $c' = c^{(2-d)/(d-1)}$.
Let $\alpha = n/k$ be the inverse yield in the small input error rate limit.
Clearly, $C \ge \alpha^m$.
Since $q(1)=1$, the probability that the output is the desired magic state
can be at most $1-p_0^C$.
It follows that $p_0 ^C \le \epsilon$, and therefore, $\alpha \ge d$.
We conclude that
\[
C \ge d^m = \Omega( \log (1/\epsilon) ).
\]

\section{A family of triorthogonal matrices}
\label{sec:family}

To construct explicit distillation protocols,
triorthogonal matrices $G$ with high yield $k/n$ are called for.
A natural strategy to maximize the yield is to keep the number
of even-weight rows in $G$ as small as possible.
 Indeed, each extra row in $G_0$
increases the number of  constraints due to Eqs.~(\ref{ort2},\ref{ort3})
without increasing the yield.
However, the number of rows in $G_0$ cannot be too small.
Recall that the distillation subroutine of Section~\ref{sec:dist}
improves the quality of magic states only if $d\ge 2$, where
$d$ is the distance of the code $\css{X,\calG_0;Z,\calG^\perp}$
against $Z$-errors defined in Eq.~(\ref{distance}).
We claim that $d=1$ whenever $G_0$ has less than three rows.
Indeed, in this case Lemma~\ref{lemma:3rows} implies that
$G_0$ must have a zero column, say, the first one.
Then $e_1\equiv (1,0,\ldots,0)\in \calG_0^\perp$. On the other hand,
$e_1\notin \calG^\perp$ since otherwise the first column of $G$ would be zero.
It shows that
$d=1$, see Eq.~(\ref{distance}). Hence a good strategy is
to look for candidate triorthogonal matrices with $3$ even-weight rows
such that $G_0$ has no zero columns. This guarantees $d\ge 2$.

Below  we present a family of triorthogonal matrices with yield $k/n = k/(3k+8)$
where $k$ is even.
The matrices are constructed from several simple submatrices, which we define first:
\begin{align}
L = \begin{bmatrix}
1 & 1 & 1 & 1 \\
1 & 1 & 1 & 1 \\
\end{bmatrix},\quad
& 
M = \begin{bmatrix}
1 & 1 & 1 & 0 & 0 & 0 \\
0 & 0 & 0 & 1 & 1 & 1 \\
\end{bmatrix}, \nonumber \\
S_1 = \begin{bmatrix}
0 & 1 & 0 & 1 \\
0 & 0 & 1 & 1 \\
1 & 1 & 1 & 1 \\
\end{bmatrix},\quad
& 
S_2 = \begin{bmatrix}
1 & 0 & 1 & 1 & 0 & 1 \\
0 & 1 & 1 & 0 & 1 & 1 \\
0 & 0 & 0 & 0 & 0 & 0 \\
\end{bmatrix}.
\end{align}
For each even number $k \ge 0$, define $(k+3) \times (3k+8)$ matrix
\begin{align}
G(k) =
\begin{bmatrix}
  0   & L    & M    & 0   & \cdots & 0   \\
  0   & L    & 0    & M   &        & 0   \\
\vdots&\vdots&\vdots&     & \ddots & 0   \\
  0   & L    & 0    & 0   & \cdots & M   \\
S_1   & S_1  & S_2 & S_2 & \cdots & S_2 \\
\end{bmatrix},
\end{align}
where $L,M$, and $S_2$ respectively appear $k/2$ times.

This family of matrices is triorthogonal with $k$ odd-weight rows
and $3$ even-weight rows.
To see this, first consider the usual orthogonality condition Eq.~\eqref{ort2}.
Any pair of rows from $G(k)_1$, the upper $k$ rows, overlap in $L$, which has weight 4.
The bottom three rows, $G(k)_0$ give three pairs whose overlaps have weight $4,4$, and $2+k$, respectively.
A row from $G(k)_1$ and another from $G(k)_0$ overlap at 4 positions.
Thus, the rows of $G(k)$ are mutually orthogonal.
One can check similarly the triorthogonality condition Eq.~\eqref{ort3}.

For any linear space $\calF\subseteq \FF_2^n$ define its
weight enumerator as $W_\calF(x)=\sum_{f\in \calF} x^{|f|}$.
The error analysis in Section~\ref{sec:dist} requires the weight enumerators
of $\calG(k)_0$
and $\calG(k)_0 \oplus (g^a)$ for all $a = 1,\ldots, k$, where $g^a$ are the rows of $G(k)_1$.
Due to the periodic structure of $G(k)$,
the weight enumerator of $\calG(k)_0 \oplus (g^a)$ is independent of $a$.
The classical codes $\calG(k)_0$ and $\calG(k)_0 \oplus (g^1)$ have only 8 and 16 code vectors, respectively,
and therefore an explicit calculation is easy:
\begin{align}
 W_{\calG(k)_0}(x) &= 1 + x^8 + 6x^{4+2k} \\
W_{\calG(k)_0 \oplus (g^1)}(x) &= 1+ 2 x^7 + x^8 + 6 x^{3 + 2 k} + 6 x^{4+2k} \nonumber
\end{align}

If $G(k)$ is used in our distillation protocol,
the success probability or acceptance rate given the input error rate $p$
is
\[
 P_s(p) = 1 - (8+3k) p + \cdots,
\]
and the output error rate $q$ on any \emph{one} qubit is
\[
 q(p) = (1+3k)p^2 + \cdots
\]
by Eq.~\eqref{qdual}, where $\cdots$ indicate higher order terms in $p$.
The initial term of $q(p)$ can be intuitively understood.
Since the stabilizer code $\css{X,\calG(k)_0;Z,\calG(k)^\perp}$ has logical $Z$ operators of weight 2,
the probability that there is an undetected error on the output qubit is $O(p^2)$.
The coefficient of $p^2$ is the number of logical $Z$ operators of weight 2
that acts nontrivially on a particular logical qubit,
which is readily counted as $4+3(k-1)$.

The threshold input error rate can be obtained by the requirement that $q(p) < p$.
From the leading term of $q(p)$, one may estimate the threshold as
\[
p_{\text{th}} \approx \frac{1}{3k+1}.
\]
Provided that the input error rate is smaller than $p_\text{th}$,
solving Eq.~\eqref{eq:cost-recursion} gives
\[
 C(\epsilon) = O\left(\log^\gamma \frac{1}{\epsilon}\right), \quad \gamma = \log_2 \frac{3k+8}{k} .
\]
The scaling exponent $\gamma$ reaches $\log_2 3 \approx 1.585$ in the large $k$ limit,
which is the best to the authors' awareness.

\section{Comparison with known protocols}
\label{sec:cost}

\begin{figure}[btp]
 \includegraphics[width=.45\textwidth]{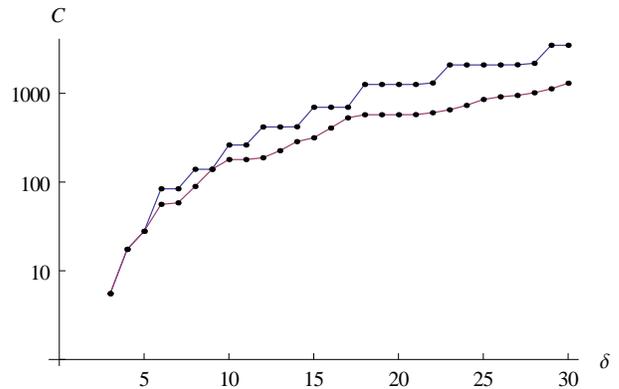}
\caption{Distillation cost $C$
as a function of the target error rate $\epsilon = 10^{-\delta}$
for a fixed input error rate $p=0.01$.
The upper curve is obtained from \cite{MEK},
and the lower curve from the optimization using the triorthogonal matrices $G(k)$.
The lines are mere guide to eyes.}
\label{fig:cost-vs-epsilon}
\end{figure}

The output error rate improves most greatly
when the input error rate is much smaller than the threshold of the protocol.
One cannot thus use $G(k)$ naively with large $k$ since the threshold is inversely proportional to $k$.
It is therefore desirable to concatenate various protocols to minimize the resource requirement.
This optimization is carried out for illustrative purpose by a numerical computation.
We restrict the number of rounds to be less than or equal to 5,
and consider all possible combinations of
\begin{enumerate}
 \item (``15'') the 15-to-1 protocol~\cite{BK04},
 \item (``5'') the 10-to-2 protocol~\cite{MEK},
 \item (``$k$'') the $(3k+8)$-to-$k$ protocol using the triorthogonal matrices $G(k)$ for $k = 2,4,6,\ldots,40$, and
 \item (``49'') the 49-to-1 protocol presented in Appendix~\ref{app49}.
\end{enumerate}
The result is summarized in Table~\ref{tb:cost},
where the numbers in the parenthesis above are used to denote each subroutine.
Unfortunately, ``49'' in the optimization
had found no place in the best combinations.
See also Fig.~\ref{fig:cost-vs-epsilon}.

A general rule is that it is better to use high threshold protocols for initial rounds,
and then use high yield protocols when the error rate becomes small.

\begin{table}[tbp]
\begin{tabular}{c|c|c|c|c}
\hline
\hline
$-\log_{10} \epsilon_\text{target}$ & Protocol  &  $-\log_{10} \epsilon_\text{actual}$ & $C$ & $C_\text{MEK}$ \\
\hline
 3 & 5 & 3.030 & 5.521 & 5.521 \\
 4 & 15 & 4.443 & 17.44 & 17.44 \\
 5 & \text{5-5} & 5.104 & 27.86 & 27.86 \\
 6 & \text{15-40} & 6.802 & 56.07 & 83.99 \\
 7 & \text{15-24} & 7.022 & 58.30 & 83.99 \\
 8 & \text{5-5-40} & 8.125 & 89.26 & 139.3 \\
 9 & \text{5-5-5} & 9.253 & 139.3 & 139.3 \\
 10 & \text{15-40-40} & 11.52 & 179.4 & 261.7 \\
 11 & \text{15-40-40} & 11.52 & 179.4 & 261.7 \\
 12 & \text{15-24-36} & 12.01 & 187.9 & 418.0 \\
 13 & \text{15-10-20} & 13.00 & 225.6 & 418.0 \\
 14 & \text{5-5-40-40} & 14.17 & 285.6 & 419.9 \\
 15 & \text{5-5-18-28} & 15.00 & 315.5 & 696.7 \\
 16 & \text{5-5-6-22} & 16.03 & 406.2 & 696.7 \\
 17 & \text{5-5-5-10} & 17.02 & 529.5 & 696.7 \\
 18 & \text{15-40-40-40} & 20.96 & 574.1 & 1260. \\
 19 & \text{15-40-40-40} & 20.96 & 574.1 & 1260. \\
 20 & \text{15-40-40-40} & 20.96 & 574.1 & 1260. \\
 21 & \text{15-38-40-40} & 21.05 & 575.9 & 1260. \\
 22 & \text{15-22-38-40} & 22.03 & 604.3 & 1308. \\
 23 & \text{15-14-30-40} & 23.01 & 652.3 & 2090. \\
 24 & \text{15-10-18-40} & 24.01 & 731.5 & 2090. \\
 25 & \text{15-6-16-36} & 25.01 & 853.1 & 2090. \\
 26 & \text{5-5-40-40-40} & 26.25 & 914.0 & 2090. \\
 27 & \text{5-5-26-38-40} & 27.04 & 947.5 & 2100. \\
 28 & \text{5-5-16-32-40} & 28.01 & 1015. & 2181. \\
 29 & \text{5-5-10-26-38} & 29.01 & 1125. & 3483. \\
 30 & \text{5-5-8-14-30} & 30.01 & 1301. & 3483. \\
\hline
\hline
\end{tabular}
\caption{
Minimum average number $C$ of required input magic states of the fixed error rate $p_\text{in}=0.01$
to distill a single output magic state of error rate $\le \epsilon_\text{target}$.
The sequence of labels in the second column denotes the subroutines in order from left to right.
An even number $k$ in the second column denotes the one round of distillation using $G(k)$.
``15'' and ``5'' respectively represent the protocol by \cite{BK04} and \cite{MEK}.
$C_\text{MEK}$ utilized only ``15'' and ``5''.
The table is numerically optimized under the restriction that there be at most 5 rounds of distillation.
}
\label{tb:cost}
\end{table}

\section{Linear equations for triorthogonal matrices}
\label{sec:linear}

The triorthogonality Eq.~(\ref{ort2},\ref{ort3}) in general depends on a particular presentation of $G$
and is not automatically guaranteed by the classical code $\calG$.
However, a certain choice of variables associated to $G$ yields a set of linear equations over $\FF_2$,
equivalent to the triorthogonality.
This system of linear equations makes numerical search effective.

Suppose a triorthogonal matrix $G$ is of size $m \times n$.
Let $x = (x_1,\ldots, x_m) \in \FF_2^m$ denote an arbitrary $m$-bit string.
Each column of the matrix $G$ corresponds to a particular $x \in \FF_2^m$;
in other words, $G$ is described by $n$ such bit strings $x$.
The cardinality of the overlap between $a$-th and $b$-th row ($a \neq b$)
is exactly the number of columns $x$ in $G$ such that $x_a = x_b = 1$.
Let $N_x$ be the number of columns $x$ appearing in $G$.
Then, the usual orthogonality condition Eq.~\eqref{ort2} can be written as
\begin{equation}
 \sum_{x \in \FF_2^m: x_a = x_b = 1} N_x = 0 \pmod 2
\label{eq:ort2linear}
\end{equation}
for distinct $a,b$.
Likewise, the cardinality of the triple overlap among distinct rows $a,b,c$ is
exactly the number of columns $x$ such that $x_a = x_b = x_c = 1$.
Therefore, the triorthogonality condition Eq.~\eqref{ort3} is equivalent to
\begin{equation}
 \sum_{x \in \FF_2^m: x_a = x_b = x_c = 1} N_x = 0 \pmod 2
\label{eq:ort3linear}
\end{equation}
for distinct $a,b,c$.
The weight of each row $a$ is the sum $\sum_{x : x_a = 1} N_x$.
Demanding $k$ odd-weight rows of $G$ is possible by the following inhomogeneous equations.
\begin{align}
\sum_{x \in \FF_2^m: x_a = 1} N_x =
\begin{cases}
 1 \pmod 2 & \text{if } 1 \le a \le k, \\
 0 \pmod 2 & \text{otherwise.}
\end{cases}
\label{eq:wt-linear}
\end{align}

Conversely, treating all $N_x$ as unknown binary variables, any solution to
Eqs.~(\ref{eq:ort2linear},\ref{eq:ort3linear},\ref{eq:wt-linear})
gives rise to a triorthogonal matrix. Namely, we just write a column $x^T=(x_1,\ldots,x_m)^T$ whenever $N_x = 1$.
The number of columns of the resulting matrix will be the Hamming weight of the vector $N$
whose components are indexed by $x \in \FF_2^m$.

One does not have to be concerned about the situation $N_x > 1$
because it only produces less efficient protocols for magic state distillation.
Suppose there are repeated columns in an $n' \times m$ triorthogonal matrix $G'$,
and let $G$ be the $n \times m$ triorthogonal matrix
obtained from $G'$ by removing repeated columns in pairs.
Consider $Z(f)$, a logical operator of $\calC' = \css{X,\calG'_0;Z,\calG'^\perp}$ of minimal weight.
The support of $f$ should not involve any pair of indices of the repeated columns due to the minimality.
Hence, $Z(f)$ may be thought of a logical operator of $\calC = \css{X,\calG_0;Z,\calG^\perp}$.
Conversely, any logical operator of $\calC$ can be viewed as that of $\calC'$.
Therefore, $\calC$ and $\calC'$ have the same minimal weight for $Z$-type logical operators,
but $\calC'$ has longer length.
For the same reason, it is safe to assume $N_{(0,0,\ldots,0)} = 0$.

The set of all solutions to the Eqs.~(\ref{eq:ort2linear},\ref{eq:ort3linear},\ref{eq:wt-linear})
contains useless triorthogonal matrices.
In order for a protocol to be useful,
the minimal weight for $Z$-type logical operators must be at least $2$.
If a triorthogonal matrix $G$ has an all zero column in $G_0$,
the lower $m-k$ even-weight rows, then the resulting stabilizer code
$\css{X,\calG_0;Z,\calG^\perp}$ admits weight one $Z$-type logical operator.
Thus, we should impose the following linear constraints.
\begin{equation}
 N_{(x_1,\ldots,x_k, 0, \ldots, 0)} = 0
\end{equation}
for all $(x_1,\ldots,x_k)\in \FF_2^k$.

So, given the number $m$ of rows of $G$ and the number $k$ of odd-weight rows,
one can solve the above equations over $\FF_2$ to find the minimal weight solution $N$.
There are $2^m$ variables $N_x$ and $2^k + \binom{m}{1} + \binom{m}{2} + \binom{m}{3}$ equations.
Note that due to Lemma~\ref{lemma:3rows}, one has to consider the case $m - k \ge 3$.

\appendix
\section{Alternative distillation subroutine}
\label{appdst}
In this section we show that the distillation scheme proposed in
Ref.~\cite{BK04} can be adapted to any stabilizer code based on
a triorthogonal matrix. It can serve as an alternative to the
subroutine described in Section~\ref{sec:dist}. Both subroutines
output the same state and have the same success probability.

Let $G$ be any triorthogonal matrix with $n$ columns, $k$ odd-weight
rows $f^1,\ldots,f^k$, and $m-k$ even-weight rows.
Consider the following distillation protocol that takes
$n$ input qubits and outputs $k$ qubits.
\begin{enumerate}
\item Measure eigenvalues of $Z(f)$, $f\in \calG^\perp$.
Let the eigenvalue of $Z(f)$ be $(-1)^{\mu(f)}$,
where $\mu\, :\, \calG^\perp \to \FF_2$ is a linear function ($Z$-syndrome).
\item  Choose any $w \in \FF_2^n$ such that $\mu(f)=(w,f)$
for all $f\in \calG^\perp$. Apply $A(w)^\dag$.
\item Apply unitary $U$ from Lemma~\ref{lemma:transversal}.
\item Measure eigenvalues of $X(g)$, $g\in \calG_0$.
Declare `FAILED' unless all eigenvalues are $+1$.
\item Decode $\css{X,\calG_0;Z,\calG^\perp}$.
\end{enumerate}
Note that the measurements of $Z(f)$ and $X(g)$ at Steps~1,4
only need to be performed for basis vectors  $f\in \calG^\perp$ and $g\in \calG_0$ respectively.
Hence the total number of Pauli measurements is
\[
\dim{(\calG^\perp)}+\dim{(\calG_0)}=(n-m)+(m-k)=n-k.
\]
Let $\rho=(1-p)\ket{A_0}\bra{A_0} + p \ket{A_1}\bra{A_1}$ be the raw ancilla.
We claim that the above protocol maps $\rho^{\otimes n}$
to the output state defined in Eqs.~(\ref{rho_out},\ref{p_out}),
while the success probability $P_s(p)$ is given by Eq.~(\ref{P_s}).
Indeed, since
the input state $\rho^{\otimes n}$ is diagonal in the $A$-basis
and the correcting operator $A(w)^\dag$
has the same $Z$-syndrome as the one measured at Step~1,
the state obtained after Step~2 is
\[
\eta_2=\Pi_Z \rho^{\otimes n} \Pi_Z/\calZ,
\]
where $\Pi_Z$ projects onto the subspace
with the trivial $Z$-syndrome
and $\calZ$ is a normalizing coefficient such that $\trace{(\eta_2)}=1$.
Since $\rho=\calE(\ket{A}\bra{A})$,
where $\calE$ involves only $Z$-errors, see Eq.~(\ref{TE}), one gets
\begin{equation}
\label{calZ}
\calZ= \bra{A^{\otimes n}} \Pi_Z \ket{A^{\otimes n}}
=\bra{ +^n} \Pi_Z \ket{+^n}.
\end{equation}
Consider a pair of codes
\[
\calC_X\equiv \css{X,\calG_0;Z,\calG^\perp} \quad \text{and} \quad
\calC_A\equiv \css{A,\calG_0;Z,\calG^\perp},
\]
where we adopt notations of Ref.~\cite{BK04}. Note that
$\calC_A$ has non-Pauli stabilizers $A(g)$, $g\in \calG_0$
in addition to Pauli ones $Z(g)$, $g\in \calG^\perp$.
By abuse of notations we shall sometimes identify $\calC_X$ and $\calC_A$
with the codespaces of the respective codes.
Taking into account that $A=TXT^\dag$ and $TZ=ZT$ we conclude that $\calC_A=\hat{T}\cdot \calC_X$, where $\hat{T}=T^{\otimes n}$.
Let $U$ be the diagonal Clifford unitary constructed in Lemma~\ref{lemma:transversal}.
From Eq.~(\ref{UT}) we infer that
$U\hat{T}$ preserves the codespace
$\calC_X$ and thus
\begin{equation}
\label{A-X}
U \cdot \calC_A = \calC_X.
\end{equation}
This shows that $\ket \psi \in \calC_A$ can be specified
by eigenvalue equations $\Pi_Z \ket \psi= \ket \psi$ and
\begin{equation}
\label{CA_stab}
U^\dag X(g) U \ket \psi = \ket \psi \quad \text{for all }  g\in \calG_0.
\end{equation}
To analyze the rest of the protocol it will be convenient to insert
two dummy steps between Step~4 and Step~5, namely,
{\em Step~4a:} Apply $U^\dag$, and {\em Step~4b:} Apply $U$.
Taking into account Eq.~(\ref{CA_stab}) we conclude that
the overall effect of Steps~1-4a is to project the state
$\rho^{\otimes n}$ onto the codespace $\calC_A$. Let $\Pi_A$ be the projector
onto the subspace with the trivial $A$-syndrome of the code $\calC_A$. Then the
(unnormalized) state obtained after Step~4a is
\[
\eta_{4a}= \Pi_Z \Pi_A \rho^{\otimes n} \Pi_A \Pi_Z/\calZ,
\]
while the success probability is determined by $P_s=\trace{(\eta_{4a})}$.
Consider any term $\Pi_A Z(f) |A^{\otimes n}\rangle$ in $\eta_{4a}$.
Since $\Pi_A |A^{\otimes n}\rangle=|A^{\otimes n}\rangle$, the state $\eta_{4a}$
gets contributions only from errors $Z(f)$ such that $\Pi_A Z(f)\Pi_A\ne 0$.
Such errors must commute with any $A$-type stabilizer which is possible only
if $f\in \calG_0^\perp$. In this case one has $\Pi_A Z(f)=Z(f)\Pi_A$.
This shows that
\[
\eta_{4a}=\frac1{\calZ} \tilde{\calE} ( \Pi_Z |A \rangle\langle A| ^{\otimes n} \Pi_Z),
\]
where $\tilde{\calE}$ is a linear map defined as
\[
\tilde{\calE}(\eta)= \sum_{f\in \calG_0^\perp} (1-p)^{n-|f|} p^{|f|} Z(f) \eta Z(f).
\]
The identity $\ket A =T \ket +$ and Lemma~\ref{lemma:transversal} yield
\[
\frac{\Pi_Z\, |A^{\otimes n}\rangle}{\sqrt{\calZ}} =
\frac{\hat{T} \Pi_Z \, |+^{\otimes n} \rangle}{\sqrt{\calZ}}= \hat{T} \, \ket G =
U^\dag |\overline{A^{\otimes k}}\rangle.
\]
Note that all states above are normalized.
Thus the state obtained after Step~4b (i.e. after Step~4 of the original protocol) is
\[
\eta_4=\tilde{\calE}(|\overline{A^{\otimes k}}\rangle \langle \overline{A^{\otimes k}}|).
\]
This shows that $P_s=\trace{(\eta_4)}$ is indeed given by Eq.~(\ref{P_s}).
As was shown in Section~\ref{sec:dist},
decoding the state $\eta_4$ yields the desired output state Eq.~(\ref{rho_out}).

\vspace{3mm}

\section{49-to-1 protocol}
\label{app49}

The approach pursued in this paper aims at minimizing
the distillation cost scaling exponent
$\gamma = \log(n/k) / \log d$ by constructing codes
with high yield $k/n$ and $d=2$.
An alternative method of constructing codes with large distance $d$
and small yield (e.g. $k=1$) appears to be less fruitful.
Using the linear system method of Section~\ref{sec:linear}
we were able to find a $49$-qubit code with $k=1$
that admits a transversal $T$-gate and has distance
 $d=5$. The corresponding  triorthogonal matrix $G_0$
 of size $13\times 49$
 is shown below.
\begin{align*}
G_{0} =
\tiny{
\begin{bmatrix}
1111111111111110101010101010101010101010101010101 \\
0000000000000000000111100110011000011001100110011 \\
0000000000000001100000011001100110000000000000000 \\
0000000000000000000000000000000001111000000001111 \\
0000000000000000011110000000000000000111100000000 \\
0000000000000000000001111000011110000000000000000 \\
0000000000000000000000000111111110000000000000000 \\
0000000000000000000000000000000001111111100000000 \\
0000000000000000000000000000000000000000011111111 \\
1010101010101010000000000000000000000000000000000 \\
0110011001100110000000000000000000000000000000000 \\
0001111000011110000000000000000000000000000000000 \\
0000000111111110000000000000000000000000000000000 \\
\end{bmatrix}
}
\end{align*}
The weight enumerator of $\calG_0$ computed numerically is
\[
 W_{49}(x) = 1+32 x^8+442 x^{16}+6696 x^{24}+1021 x^{32} .
\]
Thus, $\calG_0$ is a triply-even linear code~\cite{BetsumiyaMunemasa2010triply},
that is, $|f|=0\pmod 8$ for any $f\in \calG_0$.
By adding all-ones row to $G_0$, one obtains a
triorthogonal matrix $G$ with $k=1$. It leads to a
protocol distilling $1$ magic state out of $49$ input states.
Note that for any triorthogonal matrix with one odd-weight row
$1^n$  the relevant distance $d$ defined in Eq.~(\ref{distance}) can be
written as
\begin{equation}
\label{distance1}
d=\min_{\substack{f\in \calG_0^\perp \\ \text{$|f|$ is odd} \\ }} \; \; |f|.
\end{equation}
We have checked numerically that $d=5$ for the $49$-qubit code.
Since the code is triply-even,
the Clifford operator $U$ defined in  Lemma~\ref{lemma:transversal} is the identity.
The output error rate as a function of input error rate has the leading term
\[
 q_{49}(p) = 1411 p^5 + \cdots .
\]
The distillation threshold was found to be $p_{49;\text{th}} = 0.1366$.

We note that the above $49$-qubit code is optimal in the sense
that there are no triply-even linear codes of odd length $n\le 47$
such that the distance $d$ defined in Eq.~(\ref{distance1})
is greater than $3$.
This fact can be checked numerically using the classification
of all maximal triply-even codes of length $48$ found in~\cite{BetsumiyaMunemasa2010triply}.
A maximal triply-even code of length $47$ or shorter
can be thought of as a subcode of some maximal triply-even code of length $48$
obtained by imposing the linear condition for one component to be zero.
Using results of~\cite{BetsumiyaMunemasa2010triply} we
were able to examine numerically
all maximal triply-even codes of length $47$.
We found that $d\le 3$ for all such codes.
Further shortening cannot increase the distance $d$.

\acknowledgements
JH is in part supported
by the Institute for Quantum Information and Matter (IQIM), an NSF Physics Frontier Center,
and by the Korea Foundation for Advanced Studies.
JH thanks the hospitality of IBM Watson Research Center,
where he was a summer intern while this work is done.
SB was partially supported by the
DARPA QUEST program under contract number HR0011-09-C-0047
and by the Intelligence Advanced Research Projects Activity (IARPA) via Department of Interior National Business Center contract number D11PC20167. The U.S. Government is authorized to reproduce and distribute reprints for Governmental purposes notwithstanding any copyright annotation thereon. Disclaimer: The views and conclusions contained herein are those of the authors and should not be interpreted as necessarily representing the official policies or endorsements, either expressed or implied, of IARPA, DoI/NBC, or the U.S. Government.


\begin{thebibliography}{24}%
\makeatletter
\providecommand \@ifxundefined [1]{%
 \@ifx{#1\undefined}
}%
\providecommand \@ifnum [1]{%
 \ifnum #1\expandafter \@firstoftwo
 \else \expandafter \@secondoftwo
 \fi
}%
\providecommand \@ifx [1]{%
 \ifx #1\expandafter \@firstoftwo
 \else \expandafter \@secondoftwo
 \fi
}%
\providecommand \natexlab [1]{#1}%
\providecommand \enquote  [1]{``#1''}%
\providecommand \bibnamefont  [1]{#1}%
\providecommand \bibfnamefont [1]{#1}%
\providecommand \citenamefont [1]{#1}%
\providecommand \href@noop [0]{\@secondoftwo}%
\providecommand \href [0]{\begingroup \@sanitize@url \@href}%
\providecommand \@href[1]{\@@startlink{#1}\@@href}%
\providecommand \@@href[1]{\endgroup#1\@@endlink}%
\providecommand \@sanitize@url [0]{\catcode `\\12\catcode `\$12\catcode
  `\&12\catcode `\#12\catcode `\^12\catcode `\_12\catcode `\%12\relax}%
\providecommand \@@startlink[1]{}%
\providecommand \@@endlink[0]{}%
\providecommand \url  [0]{\begingroup\@sanitize@url \@url }%
\providecommand \@url [1]{\endgroup\@href {#1}{\urlprefix }}%
\providecommand \urlprefix  [0]{URL }%
\providecommand \Eprint [0]{\href }%
\@ifxundefined \urlstyle {%
  \providecommand \doi  [0]{\begingroup \@sanitize@url \@doi}%
  \providecommand \@doi [1]{\endgroup \@@startlink {\doibase
  #1}doi:\discretionary {}{}{}#1\@@endlink }%
}{%
  \providecommand \doi  [0]{doi:\discretionary{}{}{}\begingroup
  \urlstyle{rm}\Url }%
}%
\providecommand \doibase [0]{http://dx.doi.org/}%
\providecommand \Doi [0]{\begingroup \@sanitize@url \@Doi }%
\providecommand \@Doi  [1]{\endgroup\@@startlink{\doibase#1}\@@Doi}%
\providecommand \@@Doi [1]{#1\@@endlink}%
\providecommand \selectlanguage [0]{\@gobble}%
\providecommand \bibinfo  [0]{\@secondoftwo}%
\providecommand \bibfield  [0]{\@secondoftwo}%
\providecommand \translation [1]{[#1]}%
\providecommand \BibitemOpen [0]{}%
\providecommand \bibitemStop [0]{}%
\providecommand \bibitemNoStop [0]{.\EOS\space}%
\providecommand \EOS [0]{\spacefactor3000\relax}%
\providecommand \BibitemShut  [1]{\csname bibitem#1\endcsname}%
\bibitem [{\citenamefont {Dennis}\ \emph {et~al.}(2002)\citenamefont {Dennis},
  \citenamefont {Kitaev}, \citenamefont {Landahl},\ and\ \citenamefont
  {Preskill}}]{Dennis01}%
  \BibitemOpen
  \bibfield  {author} {\bibinfo {author} {\bibfnamefont {E.}~\bibnamefont
  {Dennis}}, \bibinfo {author} {\bibfnamefont {A.}~\bibnamefont {Kitaev}},
  \bibinfo {author} {\bibfnamefont {A.}~\bibnamefont {Landahl}}, \ and\
  \bibinfo {author} {\bibfnamefont {J.}~\bibnamefont {Preskill}},\ }\href@noop
  {} {\bibfield  {journal} {\bibinfo  {journal} {J. Math. Phys.},\ }\textbf
  {\bibinfo {volume} {43}},\ \bibinfo {pages} {4452} (\bibinfo {year}
  {2002})}\BibitemShut {NoStop}%
\bibitem [{\citenamefont {Shor}(1996)}]{Shor96}%
  \BibitemOpen
  \bibfield  {author} {\bibinfo {author} {\bibfnamefont {P.~W.}\ \bibnamefont
  {Shor}},\ }\href@noop {} {\bibfield  {journal} {\bibinfo  {journal} {In Proc.
  37th FOCS},\ \bibinfo {pages} {56}} (\bibinfo {year} {1996})}\BibitemShut
  {NoStop}%
\bibitem [{\citenamefont {Knill}(2004)}]{Knill04}%
  \BibitemOpen
  \bibfield  {author} {\bibinfo {author} {\bibfnamefont {E.}~\bibnamefont
  {Knill}},\ }\href@noop {} {\bibfield  {journal} {\bibinfo  {journal}
  {arXiv:quant-ph/0404104}} (\bibinfo {year} {2004})}\BibitemShut {NoStop}%
\bibitem [{\citenamefont {Aliferis}\ \emph {et~al.}(2006)\citenamefont
  {Aliferis}, \citenamefont {Gottesman},\ and\ \citenamefont
  {Preskill}}]{AGP06}%
  \BibitemOpen
  \bibfield  {author} {\bibinfo {author} {\bibfnamefont {P.}~\bibnamefont
  {Aliferis}}, \bibinfo {author} {\bibfnamefont {D.}~\bibnamefont {Gottesman}},
  \ and\ \bibinfo {author} {\bibfnamefont {J.}~\bibnamefont {Preskill}},\
  }\href@noop {} {\bibfield  {journal} {\bibinfo  {journal} {Quant. Inf.
  Comput.},\ }\textbf {\bibinfo {volume} {6}},\ \bibinfo {pages} {97} (\bibinfo
  {year} {2006})}\BibitemShut {NoStop}%
\bibitem [{\citenamefont {Knill}(2005)}]{Knill05}%
  \BibitemOpen
  \bibfield  {author} {\bibinfo {author} {\bibfnamefont {E.}~\bibnamefont
  {Knill}},\ }\href@noop {} {\bibfield  {journal} {\bibinfo  {journal}
  {Nature},\ }\textbf {\bibinfo {volume} {434}},\ \bibinfo {pages} {39}
  (\bibinfo {year} {2005})}\BibitemShut {NoStop}%
\bibitem [{\citenamefont {{Raussendorf}}\ and\ \citenamefont
  {{Harrington}}(2007)}]{RH:cluster2D}%
  \BibitemOpen
  \bibfield  {author} {\bibinfo {author} {\bibfnamefont {R.}~\bibnamefont
  {{Raussendorf}}}\ and\ \bibinfo {author} {\bibfnamefont {J.}~\bibnamefont
  {{Harrington}}},\ }\href@noop {} {\bibfield  {journal} {\bibinfo  {journal}
  {Phys. Rev. Lett.},\ }\textbf {\bibinfo {volume} {98}},\ \bibinfo {pages}
  {190504} (\bibinfo {year} {2007})}\BibitemShut {NoStop}%
\bibitem [{\citenamefont {Fowler}\ \emph {et~al.}(2009)\citenamefont {Fowler},
  \citenamefont {Stephens},\ and\ \citenamefont {Groszkowski}}]{Fowler08}%
  \BibitemOpen
  \bibfield  {author} {\bibinfo {author} {\bibfnamefont {A.~G.}\ \bibnamefont
  {Fowler}}, \bibinfo {author} {\bibfnamefont {A.~M.}\ \bibnamefont
  {Stephens}}, \ and\ \bibinfo {author} {\bibfnamefont {P.}~\bibnamefont
  {Groszkowski}},\ }\href@noop {} {\bibfield  {journal} {\bibinfo  {journal}
  {Phys. Rev. A},\ }\textbf {\bibinfo {volume} {80}},\ \bibinfo {pages}
  {052312} (\bibinfo {year} {2009})}\BibitemShut {NoStop}%
\bibitem [{\citenamefont {{Raussendorf}}\ \emph {et~al.}(2007)\citenamefont
  {{Raussendorf}}, \citenamefont {{Harrington}},\ and\ \citenamefont
  {{Goyal}}}]{RHG07}%
  \BibitemOpen
  \bibfield  {author} {\bibinfo {author} {\bibfnamefont {R.}~\bibnamefont
  {{Raussendorf}}}, \bibinfo {author} {\bibfnamefont {J.}~\bibnamefont
  {{Harrington}}}, \ and\ \bibinfo {author} {\bibfnamefont {K.}~\bibnamefont
  {{Goyal}}},\ }\href@noop {} {\bibfield  {journal} {\bibinfo  {journal} {New
  J. Phys.},\ }\textbf {\bibinfo {volume} {9}},\ \bibinfo {pages} {199}
  (\bibinfo {year} {2007})}\BibitemShut {NoStop}%
\bibitem [{\citenamefont {Gottesman}()}]{Gottesman97}%
  \BibitemOpen
  \bibfield  {author} {\bibinfo {author} {\bibfnamefont {D.}~\bibnamefont
  {Gottesman}},\ }\href@noop {} {\bibinfo  {journal}
  {arXiv:quant-ph/9705052}}\BibitemShut {NoStop}%
\bibitem [{\citenamefont {{Bombin}}\ and\ \citenamefont
  {{Martin-Delgado}}(2009)}]{BMD:codedef}%
  \BibitemOpen
\bibfield  {journal} {  }\bibfield  {author} {\bibinfo {author} {\bibfnamefont
  {H.}~\bibnamefont {{Bombin}}}\ and\ \bibinfo {author} {\bibfnamefont {M.~A.}\
  \bibnamefont {{Martin-Delgado}}},\ }\href@noop {} {\bibfield  {journal}
  {\bibinfo  {journal} {J. Phys. A: Math. Theor.},\ }\textbf {\bibinfo {volume}
  {42}},\ \bibinfo {pages} {095302} (\bibinfo {year} {2009})}\BibitemShut
  {NoStop}%
\bibitem [{\citenamefont {Eastin}\ and\ \citenamefont
  {Knill}(2009)}]{EastinKnill2009}%
  \BibitemOpen
  \bibfield  {author} {\bibinfo {author} {\bibfnamefont {B.}~\bibnamefont
  {Eastin}}\ and\ \bibinfo {author} {\bibfnamefont {E.}~\bibnamefont {Knill}},\
  }\href@noop {} {\bibfield  {journal} {\bibinfo  {journal} {Phys. Rev.
  Lett.},\ }\textbf {\bibinfo {volume} {102}},\ \bibinfo {pages} {110502}
  (\bibinfo {year} {2009})}\BibitemShut {NoStop}%
\bibitem [{\citenamefont {Bravyi}\ and\ \citenamefont
  {Koenig}(2012)}]{BravyiKoenig12}%
  \BibitemOpen
  \bibfield  {author} {\bibinfo {author} {\bibfnamefont {S.}~\bibnamefont
  {Bravyi}}\ and\ \bibinfo {author} {\bibfnamefont {R.}~\bibnamefont
  {Koenig}},\ }\href@noop {} {\bibfield  {journal} {\bibinfo  {journal}
  {arXiv:1206.1609}} (\bibinfo {year} {2012})}\BibitemShut {NoStop}%
\bibitem [{\citenamefont {Bravyi}\ and\ \citenamefont {Kitaev}(2005)}]{BK04}%
  \BibitemOpen
  \bibfield  {author} {\bibinfo {author} {\bibfnamefont {S.}~\bibnamefont
  {Bravyi}}\ and\ \bibinfo {author} {\bibfnamefont {A.}~\bibnamefont
  {Kitaev}},\ }\href@noop {} {\bibfield  {journal} {\bibinfo  {journal} {Phys.
  Rev. A},\ }\textbf {\bibinfo {volume} {71}},\ \bibinfo {pages} {022316}
  (\bibinfo {year} {2005})}\BibitemShut {NoStop}%
\bibitem [{\citenamefont {Nebe}\ \emph {et~al.}(2000)\citenamefont {Nebe},
  \citenamefont {Rains},\ and\ \citenamefont {Sloane}}]{Nebe00}%
  \BibitemOpen
  \bibfield  {author} {\bibinfo {author} {\bibfnamefont {G.}~\bibnamefont
  {Nebe}}, \bibinfo {author} {\bibfnamefont {E.~M.}\ \bibnamefont {Rains}}, \
  and\ \bibinfo {author} {\bibfnamefont {N.~J.~A.}\ \bibnamefont {Sloane}},\
  }\href@noop {} {\bibfield  {journal} {\bibinfo  {journal}
  {arXiv:math/0001038}} (\bibinfo {year} {2000})}\BibitemShut {NoStop}%
\bibitem [{\citenamefont {Boykin}\ \emph {et~al.}(2000)\citenamefont {Boykin},
  \citenamefont {Mor}, \citenamefont {Pulver}, \citenamefont {Roychowdhury},\
  and\ \citenamefont {Vatan}}]{Boykin00}%
  \BibitemOpen
  \bibfield  {author} {\bibinfo {author} {\bibfnamefont {P.~O.}\ \bibnamefont
  {Boykin}}, \bibinfo {author} {\bibfnamefont {T.}~\bibnamefont {Mor}},
  \bibinfo {author} {\bibfnamefont {M.}~\bibnamefont {Pulver}}, \bibinfo
  {author} {\bibfnamefont {V.~P.}\ \bibnamefont {Roychowdhury}}, \ and\
  \bibinfo {author} {\bibfnamefont {F.}~\bibnamefont {Vatan}},\ }\href@noop {}
  {\bibfield  {journal} {\bibinfo  {journal} {Inf. Process. Lett.},\ }\textbf
  {\bibinfo {volume} {75}},\ \bibinfo {pages} {101} (\bibinfo {year}
  {2000})}\BibitemShut {NoStop}%
\bibitem [{\citenamefont {Reichardt}(2005)}]{Reichardt05}%
  \BibitemOpen
  \bibfield  {author} {\bibinfo {author} {\bibfnamefont {B.~W.}\ \bibnamefont
  {Reichardt}},\ }\href@noop {} {\bibfield  {journal} {\bibinfo  {journal}
  {Quant. Inf. Proc.},\ }\textbf {\bibinfo {volume} {4}},\ \bibinfo {pages}
  {251} (\bibinfo {year} {2005})}\BibitemShut {NoStop}%
\bibitem [{\citenamefont {Meier}\ \emph {et~al.}()\citenamefont {Meier},
  \citenamefont {Eastin},\ and\ \citenamefont {Knill}}]{MEK}%
  \BibitemOpen
  \bibfield  {author} {\bibinfo {author} {\bibfnamefont {A.~M.}\ \bibnamefont
  {Meier}}, \bibinfo {author} {\bibfnamefont {B.}~\bibnamefont {Eastin}}, \
  and\ \bibinfo {author} {\bibfnamefont {E.}~\bibnamefont {Knill}},\
  }\href@noop {} {\bibinfo  {journal} {arXiv:1204.4221}}\BibitemShut {NoStop}%
\bibitem [{\citenamefont {Campbell}\ \emph {et~al.}(2012)\citenamefont
  {Campbell}, \citenamefont {Anwar},\ and\ \citenamefont
  {Browne}}]{Campbell12}%
  \BibitemOpen
\bibfield  {journal} {  }\bibfield  {author} {\bibinfo {author} {\bibfnamefont
  {E.~T.}\ \bibnamefont {Campbell}}, \bibinfo {author} {\bibfnamefont
  {H.}~\bibnamefont {Anwar}}, \ and\ \bibinfo {author} {\bibfnamefont {D.~E.}\
  \bibnamefont {Browne}},\ }\href@noop {} {\bibfield  {journal} {\bibinfo
  {journal} {arXiv:1205.3104}} (\bibinfo {year} {2012})}\BibitemShut {NoStop}%
\bibitem [{\citenamefont {Jochym-O'Connor}\ \emph {et~al.}()\citenamefont
  {Jochym-O'Connor}, \citenamefont {Yu}, \citenamefont {Helou},\ and\
  \citenamefont {Laflamme}}]{Laflamme12}%
  \BibitemOpen
  \bibfield  {author} {\bibinfo {author} {\bibfnamefont {T.}~\bibnamefont
  {Jochym-O'Connor}}, \bibinfo {author} {\bibfnamefont {Y.}~\bibnamefont {Yu}},
  \bibinfo {author} {\bibfnamefont {B.}~\bibnamefont {Helou}}, \ and\ \bibinfo
  {author} {\bibfnamefont {R.}~\bibnamefont {Laflamme}},\ }\href@noop {}
  {\bibinfo  {journal} {arXiv:1205.6715}}\BibitemShut {NoStop}%
\bibitem [{\citenamefont {Nielsen}\ and\ \citenamefont
  {Chuang}(2000)}]{NCbook}%
  \BibitemOpen
\bibfield  {journal} {  }\bibfield  {author} {\bibinfo {author} {\bibfnamefont
  {M.~A.}\ \bibnamefont {Nielsen}}\ and\ \bibinfo {author} {\bibfnamefont
  {I.~L.}\ \bibnamefont {Chuang}},\ }\href@noop {} {\emph {\bibinfo {title}
  {Quantum Computation and Quantum Information}}}\ (\bibinfo  {publisher}
  {Cambridge University Press},\ \bibinfo {year} {2000})\BibitemShut {NoStop}%
\bibitem [{\citenamefont {Calderbank}\ and\ \citenamefont {Shor}(1996)}]{CSS1}%
  \BibitemOpen
  \bibfield  {author} {\bibinfo {author} {\bibfnamefont {A.}~\bibnamefont
  {Calderbank}}\ and\ \bibinfo {author} {\bibfnamefont {P.}~\bibnamefont
  {Shor}},\ }\href@noop {} {\bibfield  {journal} {\bibinfo  {journal} {Phys.
  Rev. A},\ }\textbf {\bibinfo {volume} {54}},\ \bibinfo {pages} {1098}
  (\bibinfo {year} {1996})}\BibitemShut {NoStop}%
\bibitem [{\citenamefont {Steane}(1996)}]{CSS2}%
  \BibitemOpen
  \bibfield  {author} {\bibinfo {author} {\bibfnamefont {A.}~\bibnamefont
  {Steane}},\ }\href@noop {} {\bibfield  {journal} {\bibinfo  {journal} {Proc.
  Roy. Soc. London A},\ }\textbf {\bibinfo {volume} {452}},\ \bibinfo {pages}
  {2551} (\bibinfo {year} {1996})}\BibitemShut {NoStop}%
\bibitem [{\citenamefont {MacWilliams}\ and\ \citenamefont
  {Sloane}(1983)}]{macslo}%
  \BibitemOpen
  \bibfield  {author} {\bibinfo {author} {\bibfnamefont {F.~J.}\ \bibnamefont
  {MacWilliams}}\ and\ \bibinfo {author} {\bibfnamefont {N.~J.~A.}\
  \bibnamefont {Sloane}},\ }\href@noop {} {\emph {\bibinfo {title} {The Theory
  of Error-Correcting Codes}}}\ (\bibinfo  {publisher} {Amsterdam:
  North-Holland},\ \bibinfo {year} {1983})\BibitemShut {NoStop}%
\bibitem [{\citenamefont {Betsumiya}\ and\ \citenamefont
  {Munemasa}()}]{BetsumiyaMunemasa2010triply}%
  \BibitemOpen
  \bibfield  {author} {\bibinfo {author} {\bibfnamefont {K.}~\bibnamefont
  {Betsumiya}}\ and\ \bibinfo {author} {\bibfnamefont {A.}~\bibnamefont
  {Munemasa}},\ }\href@noop {} {\bibinfo  {journal}
  {arXiv:1012.4134}}\BibitemShut {NoStop}%
\end{thebibliography}
%

\end{document}